\documentclass{article}

\usepackage{amsmath,amssymb,amsfonts,amsthm,fullpage}

\newtheorem{theorem}{Theorem}[section]   
\newtheorem{corollary}[theorem]{Corollary} 
\newtheorem{lemma}[theorem]{Lemma} 

\newcommand{\size}[1]{\left\lvert #1 \right\rvert}

\newcommand{\Order}{{\mathrm{O}}}

\newcommand{\set}[1]{\left\{ #1 \right\}}

\newcommand{\suppress}[1]{}
\newcommand{\comment}[1]{}

\newcommand{\etal}{\emph{et al.\/} }

\newcommand{\NP}{\textbf{NP}}
\newcommand{\coNP}{\textbf{co-NP}}
\newcommand{\BQP}{\textbf{BQP}}
\newcommand{\perm}{\textsc{Permutation}}
\newcommand{\usearch}{\textsc{Unique Search}}
\newcommand{\cA}{\mathcal{A}}
\newcommand{\cB}{\mathcal{B}}
\newcommand{\eps}{\varepsilon}

\begin{document}

\title{Inverting a permutation is as hard as unordered search}

\author{
Ashwin Nayak~\thanks{
Department of Combinatorics and Optimization, and Institute for
Quantum Computing, University of Waterloo, 200 University Ave.\ W.,
Waterloo, ON N2L 3G1, Canada.
E-mail: {\tt ashwin.nayak@uwaterloo.ca}.
Research supported in part by NSERC Canada, CIFAR, an ERA (Ontario),
QuantumWorks, MITACS, and ARO (USA).
A.N.\ is also Associate Faculty, Perimeter Institute for
Theoretical Physics, Waterloo, Canada.
Research at Perimeter Institute for Theoretical Physics is supported
in part by the Government of Canada through Industry Canada and by
the Province of Ontario through MRI.
}\\
University of Waterloo, and \\
Perimeter Institute for Theoretical Physics
}

\date{February~12, 2011}

\maketitle

\begin{abstract}
We show how an algorithm for the problem of inverting a permutation
may be used to design one for the problem of
unordered search (with a unique solution).
Since there is a straightforward reduction
in the reverse direction, the problems are essentially equivalent.

The reduction we present helps us bypass the hybrid argument due to
Bennett, Bernstein, Brassard, and Vazirani~(1997) and the quantum
adversary method due to Ambainis~(2002) that were earlier used to
derive lower bounds on the quantum query complexity of the problem of
inverting permutations.  It directly implies that the quantum query
complexity of the problem is asymptotically the same as that for unordered
search, namely in~$\Theta(\sqrt{n}\,)$.
\end{abstract}

\section{Introduction}
\label{sec-introduction}
  
Let~$n$ be a positive integer.  The problem~$\perm_n$ of inverting a
permutation~$\pi$ on the set~$[n] = \set{1,2,\ldots, n}$ is defined as
follows.  Given~$\pi$ in the form of an oracle, and~$n$ as input,
output ``yes'' if the pre-image~$\pi^{-1}(1)$ is even and ``no'' if it
is odd.  This is a natural decision version of the 
problem that asks us to \emph{find\/}~$\pi^{-1}(1)$.
A related problem is that of unordered search: Given a
function~$f : [n] \rightarrow \set{0,1}$ as an oracle, and~$n$ as
input, output ``yes'' if~$f^{-1}(1)$ is non-empty and ``no''
otherwise.  In other words, determine if~$f$ maps any element~$i \in
[n]$ to~$1$. In this article, we restrict ourselves to functions~$f$
which map at most one element to~$1$. As we might expect, these
constitute the hardest instances of unordered search. We refer to the
corresponding sub-problem as~$\usearch_n$. 

The two problems were originally used by Bennett, Brassard, Bernstein,
and Vazirani~\cite{BennettBBV97} to show limitations of quantum
computers.  The search problem~$\usearch$ was used to show that
relative to a random boolean oracle~$A$, with probability~$1$, $\NP^A
\not\subseteq \BQP^A$.  The inversion problem~$\perm$ was similarly
used to show that relative to a random permutation oracle~$A$, with
probability~$1$, $\NP^A \cap \coNP^A \not\subseteq \BQP^A$. For the
first result, Bennett \etal showed that any quantum algorithm
for~$\usearch_n$ requires~$\Omega(\sqrt{n}\,)$ queries for a constant
probability of error. This involved a hybrid argument that works for
both worst-case error and distributional error under an equal mixture
of uniform distributions over ``yes'' and ``no'' instances.
The lower bound is matched by the Grover quantum
search algorithm~\cite{Grover96}, and is therefore optimal.  For the
second result, Bennett \etal used a nested hybrid argument and showed
that any quantum algorithm for the inversion problem
requires~$\Omega(\sqrt[3]{n}\,)$ queries (for constant probability of
error under the uniform distribution).  The optimal bound
of~$\Omega(\sqrt{n}\,)$ was established for worst-case query
complexity by Ambainis~\cite{Ambainis02} using the then-newly-minted
quantum adversary method.

An algorithm for unordered search (in fact, for~$\usearch_n$) may be
used to solve the inversion problem~$\perm_n$ in the obvious manner,
using at most twice the number of oracle queries. Namely, we define a
boolean function~$f$ on~$[n]$ such that~$f(i) = 1$ iff~$\pi(i) = 1$
and~$i$ is even. This function may be evaluated with one classical
query to an oracle for~$\pi$.  An additional query is used in the
quantum case to ``erase'' the answer to the first query (see the
discussion regarding quantum queries at the end of the proof of
Theorem~\ref{thm-reduction-u} in Section~\ref{sec-reduction}).  Therefore
the Grover quantum search algorithm~\cite{Grover96} solves this
problem with~$\Order(\sqrt{n}\,)$ queries to an oracle for~$\pi$.  We
describe a reduction in the reverse direction, i.e., we show how any
algorithm that solves~$\perm_n$ may be modified to
solve~$\usearch_{n/2}$ when~$n$ is even. 
The intuition behind the reduction comes from a hybrid
argument in which we consider runs of the algorithm 
on oracles not in its domain. This kind of device has been used in
numerous works on quantum query complexity to great effect.

The reduction we present is randomized and is between 
distributional versions of the
problems, with equal weight on ``yes'' and ``no'' instances, and with
uniform conditional distributions for each kind of instance.  Due to
the inherent symmetry in the two problems under consideration, the
distributional and worst-case versions of the problems are, in fact,
equivalent in query complexity (for the distributions described above;
see Lemma~\ref{lem-avg-wc} and the discussion following it).
Thus, we are also able to derive an algorithm
for~$\usearch_{n/2}$ with a bound on its worst-case error.

Let~$\mu_n$ denote the distribution obtained by taking an equal mixture
of the sole ``no'' instance and the uniform distribution over ``yes''
instances of~$\usearch_n$.
\begin{theorem}
\label{thm-reduction}
Let~$n$ be an even positive integer.  Let~$\cA$ be an algorithm
(classical or quantum) that solves~$\perm_n$ with distributional error
at most~$\eps < 1/2$ on the uniform distribution over permutations
on~$[n]$, with~$q$ queries to the permutation oracle.  Then 
there is an algorithm of the same kind as~$\cA$ (classical or quantum) 
with distributional error at most~$(1+2\eps)/4$ with respect to the 
distribution~$\mu_{n/2}$, that solves~$\usearch_{n/2}$ with at
most~$q$ queries to the search oracle in the classical case, and 
at most~$2q$ queries in the quantum case. Moreover, this may be modified
to an algorithm for~$\usearch_{n/2}$ with worst-case error at 
most~$\tfrac{1}{3-2\eps} < 1/2$ and with the same query complexity.
\end{theorem}
We emphasize that the reduction uses only the knowledge of a
bound~$\eps$ on the distributional error of~$\cA$, and not its precise value.
An algorithm for~$\perm_n$ for~$n > 2$ with distributional
error at most~$\eps$ on the uniform distribution may be used to 
solve~$\perm_{n-1}$ with the same query complexity and distributional
error at most~$\eps + 1/2n$ on the uniform distribution (see the
discussion after Lemma~\ref{lem-avg-wc}). So a reduction similar to the
one in the theorem above exists for odd~$n$ as well.

Any quantum algorithm for $\usearch_{n/2}$ with constant probability 
of error~$ < 1/2$ with respect to~$\mu_{n/2}$ requires~$\Omega(\sqrt{n}\,)$
queries~\cite{BennettBBV97}.  
Theorem~\ref{thm-reduction} therefore implies a similar lower bound
for quantum algorithms for~$\perm_n$, and an~$\Omega(n)$ query lower
bound for classical algorithms for the problem.
(The lower bound of~$\Omega(n)$ for
$\perm_n$ in the classical case is straightforward to prove directly.)
\begin{corollary}
\label{cor-lb}
Any quantum algorithm that solves~$\perm_n$ (for any integer~$n > 1$)
with constant distributional error~$\eps < 1/2$ on the uniform
distribution over permutations on~$[n]$ requires~$\Omega(\sqrt{n}\,)$
queries. Consequently, the same query lower bound holds for algorithms
for $\perm_n$ with worst-case error at most~$\eps$.
\end{corollary}
Naturally, Theorem~\ref{thm-reduction} and Corollary~\ref{cor-lb} both 
also hold for any quantum algorithm for the \emph{search\/}
version of~$\perm_n$ with the same kind of error bound as in the
statements above.

The reduction we present bypasses the hybrid argument due to Bennett,
Bernstein, Brassard, and Vazirani~\cite{BennettBBV97} and the quantum
adversary method due to Ambainis~\cite{Ambainis02}, and shows a direct
connection between inversion and search. The hybrid argument underlying
the reduction was discovered in~2004 and communicated informally to 
a few people. The reduction is written up here for wider dissemination.

\section{The reduction}
\label{sec-reduction}

We start by fixing some notation and making preliminary observations. 
Then we sketch a hybrid argument which
paves the way for the reduction (Theorem~\ref{thm-reduction-u}). 
Lemma~\ref{lem-avg-wc} derives a worst-case algorithm for $\usearch_n$ 
from an
average-case algorithm, and together with Theorem~\ref{thm-reduction-u} 
implies Theorem~\ref{thm-reduction}. We finish by sketching how 
Theorem~\ref{thm-reduction} extends to odd~$n$.

Let~$n$ be any positive integer. Define the following sets of
``no'' and ``yes'' instances of $\perm_n$:
\begin{eqnarray*}
P_0 & = & \set{ \pi \;:\; \pi \textrm{ is a permutation on } [n],
          \quad \pi^{-1}(1) \textrm{ is odd }}, \quad
          \textrm{and} \\
P_1 & = & \set{ \pi \;:\; \pi \textrm{ is a permutation on } [n],
          \quad \pi^{-1}(1) \textrm{ is even }}.
\end{eqnarray*}
We also consider oracles
that compute functions~$h : [n] \rightarrow [n]$ with a unique
collision at~$1$, with one odd and one even number in the colliding
pair. In other words, these functions~$h$ are such that there are
precisely two distinct elements~$i,j$ with the same image
under~$h$. Moreover, $h(i) = h(j) = 1$, and precisely one of~$i,j$ is
odd (and the other is even).  Let~$Q$ denote the set of all such
functions.

Consider any fixed permutation~$\pi$ on~$[n]$. Consider also the
functions in~$Q$ that differ from~$\pi$ in exactly one point. These
are functions~$h$ with a unique collision such that the collision is
at~$1$, and~$\size{\pi^{-1}(1) \cap h^{-1}(1)} = 1$.  If~$\pi \in
P_0$, then the even element that is also mapped to~$1$ by~$h$ is
precisely the one on which~$\pi$ and~$h$ differ. Similarly, if~$\pi
\in P_1$, then the odd element that is also mapped to~$1$ by~$h$ is
precisely the one on which~$\pi$ and~$h$ differ.  Let~$Q_\pi$ denote
the set of such functions~$h$.
If we pick a uniformly random permutation~$\pi \in P_0$, and then pick a
uniformly random function~$h$ in~$Q_\pi$, then~$h$ is uniformly random
in~$Q$. The same holds if we switch~$P_0$ with~$P_1$.

Given a permutation~$\pi$ on~$[n]$ with~$n$ even, and a function~$f : [n/2]
\rightarrow \set{0,1}$, we define a function~$h_{\pi,f} : [n]
\rightarrow [n]$ as follows. If~$\pi \in P_0$, for any~$i \in [n]$,
\begin{eqnarray}
\label{eqn-h0}
h_{\pi,f}(i) & = & 
\begin{cases}
  1     & \textrm{if } i \textrm{ is even and } f(i/2) = 1 \\
\pi(i)  & \textrm{if } (i \textrm{ is odd) or } 
          (i \textrm{ is even and } f(i/2) = 0).
\end{cases}
\end{eqnarray}
If~$\pi \in P_1$, for any~$i \in [n]$,
\begin{eqnarray}
\label{eqn-h1}
h_{\pi,f}(i) & = & 
\begin{cases}
  1     & \textrm{if } (i \textrm{ is odd) and } f((i+1)/2) = 1 \\
\pi(i)  & \textrm{if } (i \textrm{ is even) or } 
          (i \textrm{ is odd and } f((i+1)/2) = 0).
\end{cases}
\end{eqnarray}
The function~$h_{\pi,f}$ coincides with~$\pi$ if~$f^{-1}(1)$ is empty
and belongs to~$Q_\pi$ otherwise.

We make use of the above properties of~$P_0, P_1, Q, Q_\pi$
and~$h_{\pi,f}$ in our reduction.

The essential idea behind our reduction
is that any algorithm that distinguishes a uniformly random 
permutation in~$P_0$ from a uniformly random one in~$P_1$ necessarily 
distinguishes a random permutation from~$P_i$ from a uniformly random 
unique-collision function in~$Q$, for at least one~$i \in \set{0,1}$.
Suppose that~$i=0$. Using convexity, we may further deduce that such
an algorithm also distinguishes the permutation~$\pi \in P_0$
from a uniformly random unique-collision function~$h \in Q_\pi$, for at
least one~$\pi$.  Since any function~$h \in Q_\pi$ differs from~$\pi$
in exactly one of~$n/2$ points when~$n$ is even, this
final problem is equivalent to the problem of unique unordered search
over a domain of size~$n/2$.
The above hybrid argument suffices to prove a query lower bound
for~$\perm_n$, but is not entirely satisfactory because it corresponds
to a non-uniform reduction.
(\emph{A priori\/}, we do not know which~$i \in \set{0,1}$ and
which~$\pi \in P_i$ would give us a correct algorithm for~$\usearch_{n/2}$.)
We may however glean a reduction in the uniform
sense from this argument. The idea is to replace the choices made on
the basis of existential arguments by randomized ones. 
We therefore try distinguishing a uniformly random~$\pi \in P_0$ from a
uniformly random~$h_{\pi,f} \in Q_\pi$, and similarly with a uniformly
random permutation from~$P_1$, with equal probability. At least one of
these probabilistic attempts succeeds, and gives us a bounded-error algorithm.

Recall that~$\mu_{n}$ is the distribution which assigns 
probability~$1/2$ to the
constant function~$0$, and probability~$1/2n$ to each of the ``yes''
instances of~$\usearch_{n}$. Let~$\mu^1_{n}$ be the uniform distribution
on ``yes'' instances alone, and let~$\mu^0_{n}$ be the analogue for the
lone ``no'' instance.

\begin{theorem}
\label{thm-reduction-u}
Let~$n$ be an even positive integer.  Let~$\cA$ be an algorithm
(classical or quantum) that solves~$\perm_n$ with distributional error
at most~$\eps < 1/2$ on the uniform distribution over permutations
on~$[n]$, with~$q$ queries to the permutation oracle.  Then there is an
algorithm of the same kind as~$\cA$ with distributional error
at most~$\tfrac{1+2\eps}{4} < 1/2$ with respect to~$\mu_{n/2}$ that
solves~$\usearch_{n/2}$ with at most~$q$ queries to the search oracle
in the classical case, and at most $2q$ queries in the quantum
case. Moreover, its error probability on an input drawn from~$\mu^0_{n/2}$
is at most~$\eps$ and that for~$\mu^1_{n/2}$ is~$1/2$.
\end{theorem}
\begin{proof}
Suppose we are given an algorithm~$\cA$ as in the statement of the
theorem that takes as input an even integer~$n \geq 1$, and an
oracle~$g : [n] \rightarrow [n]$. Let~$f$ be an input oracle
for~$\usearch_{n/2}$.
Recall the definition of the function~$h_{\pi,f} : [n]
\rightarrow [n]$, where~$\pi$ is a permutation on~$[n]$, as in
Eqs.~(\ref{eqn-h0},\ref{eqn-h1}). The following reduction~$\cB$
solves~$\usearch_{n/2}$ with distributional error as claimed:
\begin{quote}
With probability~$\tfrac{1}{2}$,
pick a uniformly random permutation~$\pi \in P_0$, 
output~$\cA(n,h_{\pi,f})$ and stop.

With probability~$\tfrac{1}{2}$,
pick a uniformly random permutation~$\pi \in P_1$,
output~$\neg \cA(n,h_{\pi,f})$ and stop.
\end{quote}

We may calculate the probability of error of the algorithm~$\cB$ by
considering ``yes'' and ``no'' instances separately.  The output of~$\cB$
on the lone ``no'' instance of~$\usearch_{n/2}$ is
\[
\begin{cases}
\cA(n,\pi) & \textrm{with probability } \frac{1}{2},
             \textrm{ for a uniformly random } \pi \in P_0,
             \textrm{ and} \\
\neg \cA(n,\pi) & \textrm{with probability } \frac{1}{2}, 
                  \textrm{ for a uniformly random } \pi \in P_1 .
\end{cases}
\]
Let~$\eps_0$ be the
probability of error of~$\cA$ on a uniformly random oracle from~$P_0$,
and let~$\eps_1$ be the corresponding quantity for~$P_1$. We
have~$\eps_0 + \eps_1 \leq 2\eps$.  The probability of error of~$\cB$ on the
``no'' instance is thus bounded by~$(\eps_0 + \eps_1)/2 \leq \eps$.

The output of~$\cB$ on a uniformly random ``yes'' instance
of~$\usearch_{n/2}$ is
\[
\begin{cases}
\cA(n,h) & \textrm{with probability } \frac{1}{2},
           \textrm{ for a uniformly random } h \in Q,
           \textrm{ and} \\
\neg \cA(n,h) & \textrm{with probability } \frac{1}{2},
                \textrm{ for a uniformly random } h \in Q.
\end{cases}
\]
If~$p$ denotes the
probability that the output~$\cA(n,h)$ is ``no'' for a uniformly
random~$h \in Q$, the probability of error of~$\cB$ on a uniformly random
``yes''  instance of~$\usearch_{n/2}$ is~$p/2 + (1-p)/2 = 1/2$.

Since~$\pi$ is known explicitly in the algorithm~$\cB$,
the function~$h_{\pi,f}$ can be evaluated with at most one query
to~$f$ in the classical case.
In the quantum case, the situation is not as straightforward.
Data in quantum memory generated during computations
may prevent the quantum interference necessary for correct working of an
algorithm. This may be the case with the answers to oracle queries
to~$f$ that are used to compute~$h_{\pi,f}$. It is therefore important
to simulate queries to~$h_{\pi,f}$ cleanly, i.e., without modifying 
any part of the workspace except the answer register.
With a standard implementation of the oracle to~$f$ as a unitary
operator, a query to~$h_{\pi,f}$ can be simulated cleanly
by using \emph{two\/} queries (see~\cite[Section~4]{BennettBBV97}
or~\cite[Section~1.5]{KayeLM07}).
We note that by working with a different implementation of the
oracle for~$f$, it is possible to simulate a query to~$h_{\pi,f}$ 
cleanly with only one query. As this leads to only a minor improvement 
in our reduction, we omit the details.
\end{proof}

Any algorithm with worst-case error~$\eps$ implies an algorithm with
distributional error~$\eps$ with respect to any distribution.  We
point out that $\usearch$ and $\perm$ both admit random
self-reductions, so the average-case algorithm described in
Theorem~\ref{thm-reduction-u} implies the worst-case algorithm claimed
in Theorem~\ref{thm-reduction}.
\begin{lemma}
\label{lem-avg-wc}
Suppose~$\cB$ is an algorithm for~$\usearch_n$ with distributional
error at most~$\eps_1$ on distribution~$\mu^1_{n}$, and at
most~$\eps_0$ on~$\mu^0_n$ such that~$\eps_0+\eps_1 < 1$.
Then, there is an algorithm for $\usearch_n$
that makes the same number of queries as~$\cB$, and has
worst-case error at most
\[
\eps \quad = \quad \frac{\max\set{\eps_0,\eps_1}}{1 + \size{\eps_0 - \eps_1}} \quad < \quad
\frac{1}{2} \enspace.
\]
\end{lemma}
\begin{proof}
Composing a function~$f : [n] \rightarrow \set{0,1}$ with a
permutation~$\sigma$ on~$[n]$ preserves the ``yes'' and the ``no''
instances of~$\usearch_n$.  Consider the algorithm~$\cB_{\mathrm{sym}}$:
\begin{quote}
Pick a uniformly random permutation~$\sigma$ on~$[n]$, and then return
the value given by the algorithm~$\cB$ with every oracle 
query~$i \in [n]$ replaced by a query to~$\sigma(i)$.
\end{quote}
Effectively, any single instance of~$\usearch_n$ is mapped to a
uniformly random instance with the same answer.  So the worst-case error
on the ``no'' instance is at most~$\eps_0$, and on any ``yes'' instance
is at most~$\eps_1$.

It only remains to equalize the bounds on error on the two
kinds of instance, and this may be accomplished by a standard
modification: depending on whether~$\eps_0 < \eps_1$ or not, we accept or
reject with some probability~$p$, and run the
algorithm~$\cB_{\mathrm{sym}}$ with probability~$1-p$. The choice of~$p
= \size{\eps_1-\eps_0}/(1+\size{\eps_1-\eps_0})$ gives us the claimed 
error bound.
\comment{
For simplicity, assume that the error on ``yes'' instances is larger: 
$\eps_0 < \eps_1$.  The
worst-case algorithm for~$\usearch_n$ is as follows. 
\begin{quote}
With probability~$p = (\eps_1-\eps_0)/(1+\eps_1-\eps_0)$, output
``yes'', and stop.

With probability~$1-p$, run the symmetrized
algorithm~$\cB_{\mathrm{sym}}$.
\end{quote}
A straightforward
calculation reveals that the worst-case error probability is as stated
in the lemma.}
\end{proof}
There is a similar randomized self-reduction for $\perm_n$.
Let~$\omega,\sigma$ be uniformly random permutations on~$[n]$ such 
that~$\omega$ maps~1 to itself, and~$\sigma$ permutes odd integers 
among themselves and even integers among themselves. Then for any
permutation~$\pi$ on~$[n]$, the composition~$\omega \circ \pi \circ \sigma$
is uniformly distributed in~$P_1$ if~$\pi$ is a ``yes'' instance, and
in~$P_0$ if it is a ``no'' instance. Therefore, an algorithm~$\cA$ for 
$\perm_n$ with distributional error at most~$\eps_1,\eps_0$ on 
uniformly random ``yes'' and ``no'' instances, respectively, 
leads to a worst-case algorithm which makes error at most~$\eps_1$ on
any ``yes'' instance and at most~$\eps_0$ on any ``no'' instance.
When~$n$ is even, the distributional error of~$\cA$ on the uniform
distribution is at most~$\eps = (\eps_0+\eps_1)/2$ and when~$n$ is odd,
it is at most~$\eps = ((1+1/n)\eps_0 + (1-1/n)\eps_1)/2$. We may now
rebalance the bound on error in the two kinds of instance to obtain a 
worst-case algorithm. The self-reduction also allows us to use an 
algorithm~$\cA$ for $\perm_n$ to solve $\perm_{n-1}$ as claimed in
Section~\ref{sec-introduction}. We first symmetrize the algorithm~$\cA$
through the self-reduction, and then run it on the extension of the input
permutation~$\pi$ on~$[n-1]$ obtained by defining~$\pi(n) = n$.

\subsection*{Acknowledgements}

I thank the anonymous \emph{Theory of Computing\/} referees for their 
insightful comments and suggestions. They vastly improved the quality of 
the presentation.

\bibliography{perm}

\begin{thebibliography}{1}

\bibitem{Ambainis02}
Andris Ambainis.
\newblock Quantum lower bounds by quantum arguments.
\newblock {\em Journal of Computer and System Sciences}, 64(4):749--898, 2002.

\bibitem{BennettBBV97}
Charles~H. Bennett, Ethan Bernstein, Gilles Brassard, and Umesh Vazirani.
\newblock Strengths and weaknesses of quantum computing.
\newblock {\em SIAM Journal on Computing}, 26(5):1510--1523, 1997.

\bibitem{Grover96}
Lov~K. Grover.
\newblock A fast quantum mechanical algorithm for database search.
\newblock In {\em Proceedings of the Twenty-Eighth Annual {ACM} Symposium on
  the Theory of Computing}, pages 212--219. ACM Press, New York, NY, USA,
  22--24 May 1996.

\bibitem{KayeLM07}
Phillip Kaye, Raymond Laflamme, and Michele Mosca.
\newblock {\em An Introduction to Quantum Computing}.
\newblock Oxford University Press, New York, NY, USA, 2007.

\end{thebibliography}

\end{document}